\documentclass[10pt, conference, compsocconf]{IEEEtran}

\usepackage{amsmath}
\usepackage{amsfonts}
\usepackage{amssymb}
\usepackage{url}
\usepackage{cite}
\usepackage{epsfig}
\usepackage{xspace}
\usepackage{graphicx}
\usepackage{subfigure}
\usepackage{wrapfig}
\usepackage{color}
\usepackage{listings}
\usepackage{stmaryrd}
\usepackage{xspace}

\usepackage{algorithm}
\usepackage{algorithmic}
\usepackage{multicol}

\newif\ifjournal\journalfalse
\newif\ifcamready\camreadytrue

\newtheorem{definition}{Definition}
\newtheorem{theorem}{Theorem}
\newtheorem{remark}{Remark}
\newtheorem{observation}{Observation}
\newtheorem{proof}{Proof}
\newtheorem{corollary}{Corollary}
\newtheorem{lemma}{Lemma}

\newcommand{\br}[1]{\ensuremath{\langle#1\rangle}\xspace}

\lstdefinelanguage{stackalg}{
  morekeywords={while, proc, function, else, begin, do, first, not, exists, then, end, if, each, in,  pop, push,  done, for
 , create,},
  sensitive=true,
  morecomment=[l]{//},
  morecomment=[s]{/*}{*/},
}

\lstdefinelanguage{ncode}{%
  morekeywords={ send, receive, future, halt, create, wait, mode, soft, hard},
}

\renewcommand{\baselinestretch}{0.955}

\newcommand{\ar}{auditable restoration\xspace}
\newcommand{\Ar}{Auditable restoration\xspace}
\newcommand{\AR}{Auditable Restoration\xspace}
\newcommand{\name}{AR}
\newcommand{\bname}{BAR}
\newcommand{\Gd}{lc}

\newcommand{\au}{auditable event\xspace}
\newcommand{\aus}{auditable events\xspace}

\newcommand{\restore}{restore\xspace}
\newcommand{\stable}{stable\xspace}
\newcommand{\restoration}{restoration\xspace}
\newcommand{\leader}{leader\xspace}

\begin{document}

\title{\AR of Distributed Programs}

\author{\IEEEauthorblockN{Reza Hajisheykhi}
\IEEEauthorblockA{Computer Science and\\Engineering Department\\
Michigan State University\\
East Lansing, MI, USA\\
Email: hajishey@cse.msu.edu}
\and
\IEEEauthorblockN{Mohammad Roohitavaf}
\IEEEauthorblockA{Computer Science and\\Engineering Department\\
Michigan State University\\
East Lansing, MI, USA\\
Email: roohitav@cse.msu.edu}
\and
\IEEEauthorblockN{Sandeep Kulkarni}
\IEEEauthorblockA{Computer Science and\\Engineering Department\\
Michigan State University\\
East Lansing, MI, USA\\
Email: sandeep@cse.msu.edu}}

\maketitle
\begin{abstract}
We focus on a protocol for auditable restoration of distributed systems. The need for such protocol arises due to conflicting requirements (e.g., access to the system should be restricted but emergency access should be provided). One can design such systems with a tamper detection approach (based on the intuition of {\em break the glass door}). However, in a distributed system, such tampering, which are denoted as auditable events, is visible only for a single node. This is unacceptable since the actions they take in these situations can be different than those in the normal mode. Moreover, eventually, the auditable event needs to be cleared so that system resumes the normal operation.

With this motivation, in this paper, we present a protocol for auditable restoration, where any process can potentially identify an auditable event. Whenever a new auditable event occurs, the system must reach an {\em auditable state} where every process is aware of the auditable event. Only after the system reaches an auditable state, it can begin the operation of restoration. Although any process can observe an auditable event, we require that only {\em authorized} processes can begin the task of restoration. Moreover, these processes can begin the restoration only when the system is in an auditable state. Our protocol is self-stabilizing and has bounded state space. It can effectively handle the case where faults or auditable events occur during the restoration protocol. Moreover, it can be used to provide auditable restoration to other distributed protocol.
\end{abstract}

\begin{IEEEkeywords}
Self-stabilization, reactive systems, adversary, formal methods
\end{IEEEkeywords} 
\section{Introduction}
\label{sec:intro}

\subsection{A Brief History and the Need for Auditable Restoration}
\label{sec:history}

Fault-tolerance focuses on the problem of what happens if the program is perturbed by undesired perturbations. In other words, it focuses on what happens if the program is perturbed beyond its legitimate states (a.k.a. {\em invariant}). There have been substantial ad-hoc {\em operational} approaches --designed for specific types of faults-- for providing fault-tolerance. For example, the idea of recovery blocks \cite{recBlock74} introduced the notion of acceptance conditions that should be satisfied at certain points in the computation. If these conditions are not satisfied, the program is restored to previous state from where another recovery block is executed. Checkpointing and recovery based approaches provide mechanism to restore the program to a previous checkpoint.
Dijkstra \cite{dij} introduced an approach for {\em specifying} (i.e., identifying what the program should provide irrespective of how it is achieved) one-type of fault-tolerance, namely stabilization \cite{dij}. A stabilizing program partitioned the state space of the program into legitimate states (predicate $S$ in Figure \ref{fig:dij}) and other states. It is required that (1) starting from any state in $S$, the program always stays in $S$ and (2) starting from any state in $\neg S$, the program recovers to a state in $S$.

Although stabilization is desirable for many programs, it is not suitable for some programs. For example, it may be impossible to provide recovery from all possible states in $\neg S$. Also, it may be desirable to satisfy certain safety properties during recovery. Arora and Gouda \cite{ag93} introduced another approach for formalizing a fault-tolerant system that ensures {\em convergence} in the presence of transient faults (e.g., soft errors, loss of coordination, bad initialization), say $f$. That is, from any state/configuration, the system recovers to its invariant $S$, in a finite number of steps. Moreover, from its invariant, the executions of the system satisfy its specifications and remain in the invariant; i.e., {\em closure}. They distinguish two types of fault-tolerant systems: {\em masking} and {\em nonmasking}. In the former the effects of failure is completely invisible to the application. In other words, the invariant is equal to the fault-span ($S=FS$ in Figure \ref{fig:mask}). In the latter, the fault-affected system may violate the invariant but the continued execution of the system yields a state where the invariant is satisfied (See Figure \ref{fig:mask}).

The approach by Arora and Gouda intuitively requires that after faults stop occurring, the program provides the {\em original functionality}. However, in some cases, restoring the program to {\em original} legitimate states so that it satisfies the subsequent specification may be impossible. Such a concept has been considered in \cite{tpdsHerlihy91} where authors introduce the notion of {\em graceful degradation}. In graceful degradation (cf. Figure \ref{fig:deg}), a system satisfies its original specification when no faults have occurred. After occurrence of faults (and when faults stop occurring), it may not restore itself to the original legitimate states ($S$ in Figure \ref{fig:deg}) but rather to a larger set of states ($S'$ in Figure \ref{fig:deg}) from where it satisfies a weaker specification. In other words, in this case, the system may not satisfy the original specification even after faults have stopped.

In some instances, especially where the perturbations are {\em security related}, it is not sufficient to restore the program to its original (or somewhat degraded) behavior. The notion of multi-phase recovery was introduced for such programs \cite{bkfm09}. Specifically, in these programs, it is necessary that recovery is accomplished in a sequence of phases, each ensuring that the program satisfies certain properties. One of the properties of interest is {\em strict 2-phase recovery}, where the program first recovers to states $Q$ that are strictly disjoint from legitimate states. Subsequently, it recovers to legitimate states (See Figure \ref{fig:2ph}).

The goal of auditable restoration is motivated by combining the principles of the strict 2-phase recovery and the principles of fault-tolerance. Intuitively, the goal of auditable restoration is to classify system perturbations into faults and auditable events, provide fault-tolerance (similar to that in Figure \ref{fig:mask}) to the faults and ensure that strict 2-phase recovery is provided for auditable events. Unfortunately, this cannot be achieved for arbitrary auditable events since in \cite{bkfm09}, it has been shown that adding strict 2-phase recovery to even a centralized program is NP-complete. Our focus is on auditable events that are (immediately) detectable. Faults may or may not be detectable. Unfortunately, adding strict 2-phase recovery to a program even in centralized systems has been shown to be NP-complete.

\subsection{Goals of Auditable Restoration}

In our work, we consider the case that the system is perturbed by possible faults and possible tampering that we call as {\em auditable events}. Given that both of these are perturbations of the system, we distinguish between them based on how and why they occur. By faults, we mean events that are random in nature. These include process failure, message losses, transient faults, etc. By auditable events, we mean events that are deliberate in nature for which a detection mechanism has been created. Among other things, the need for managing such events arises due to conflicting nature of system requirements. For example, consider a requirement that states that each process in a distributed system is physically secure. This requirement may conflict with another requirement such as each node must be provided emergency access (e.g., for firefighters).
As another example, consider the requirement that each system access must be authenticated. This may conflict with the requirement for (potential) unauthorized access in a crisis. Examples of this type are well-known in the domain of power systems, medical record systems, etc., where the problem is solved by techniques such as writing down the password in a physically secure location that can be broken into during crisis or by allowing unlimited access and using logs as a deterrent for unauthorized access.
Yet another example includes services such as gmail that (can) require that everytime a user logs in, he/she can authenticate via 2-factor authentication such as user's cell phone. To deal with situations where the user may not have access to a cell phone, the user is provided with a list of `one-time' passwords and it is assumed that these will always stay in the control of the users.
However, none of these solutions are fully satisfactory.

In this paper, we focus on a solution that is motivated by the notion of 2-phase recovery. Specifically, we would like to have the following properties:

\begin{itemize}
\item We consider the case where the auditable events are immediately detectable. This is the case in all scenarios discussed above. For example, the use of `one-time' password for gmail or violation of physical security of a process is detectable. Likewise, the passwords stored in a physically secure location can be different from those used by ordinary users making them detectable.
\item We require that if some process is affected by an auditable event, then eventually all (respectively, relevant)\footnote{In this paper, for simplicity, we assume that all processes are relevant.} processes in the system are aware of this auditable event. For example, if one process is physically tampered then it will not automatically cause the tampering to be detected at other processes. This would allow the possibility that they only provide the `emergency' services and protect the more sensitive information. This detection must occur even in the presence of faults (except those that permanently fail all processes that were aware of auditable event)\footnote{If all processes that are (directly or indirectly) aware of the auditable events fail then it is impossible to distinguish this from the scenario where these processes fail before the auditable events.}. We denote such states as {\em auditable states} in that all processes are aware of a new auditable event ($S2$ in Figure \ref{fig:reset}).
\item After all processes are aware of the auditable event, there exists at least one process that can begin the task of restoring the system to a normal state, thereby clearing the auditable event. This operation may be automated or could involve human-in-the-loop. However, the system should ensure that this operation cannot be initiated until all processes are aware of the auditable event. This ensures that any time the normal operation is restored, all processes are aware of the auditable event.
\item If the operation to restore the system to normal state succeeds even if it is perturbed by faults such as failure of processes. However, if an auditable event occurs while the system is being restored to the normal state, the auditable event has a higher priority, i.e., the operation to restore to normal state would be canceled until it is initiated at a later time.
\end{itemize}

Observe that such a solution is a variation of strict 2-phase recovery. When an auditable event occurs, the system is guaranteed to recover to a state where all processes are aware of this event ($S2$ in Figure \ref{fig:reset}). And, subsequently, the system recovers to its normal legitimate states ($S1$ in Figure \ref{fig:reset}). Our solution has the following properties:

\begin{itemize}
\item Our solution is self-stabilizing. If it is perturbed to an arbitrary state, it will recover to a state from where all future auditable events will be handled correctly.
\item Our solution can be implemented with a finite state space. The total state does not increase with the length of the computation. It is well-known that achieving finite state space for self-stabilizing programs is difficult \cite{ll}.
\item Our solution ensures that if auditable events occur at multiple processes {\em simultaneously}, it will be treated as one event restoring the system to the auditable state. However, if an auditable event occurs after the system restoration to normal states has begun (or after system restoration is complete), it will be treated as a new auditable event.
\item After the occurrence of an auditable event, the system recovers to the {\em auditable state} even if it is perturbed by failure of processes, failure of channels, as well as certain transient faults.
\item No process can initiate the restoration to normal operation unless the system is in an auditable state. In other words, in Figure \ref{fig:reset}, a process cannot begin restoration to normal state unless the system was recently in $S2$.
\item After the system restores to an auditable state and some process initiates the restoration to normal operation, it completes correctly even if it is perturbed by faults such as failure of processes or channels. However, if it is perturbed by an auditable event, the system recovers to the auditable state again.
\end{itemize}

\begin{figure*}
    \begin{center}
        \subfigure[\small Stabilization]{%
            \label{fig:dij}
            \includegraphics[height=0.24\textwidth , width=0.30\textwidth]{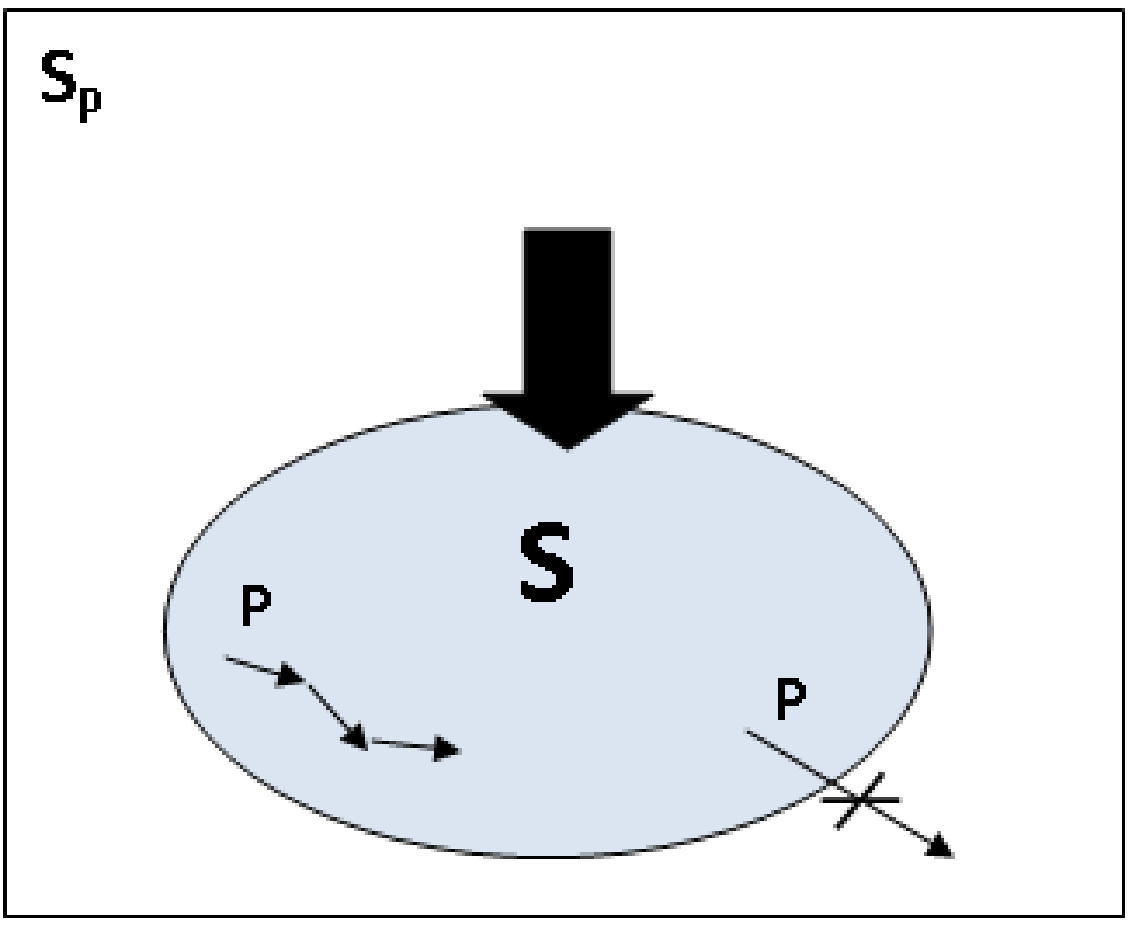}
        }%
        \subfigure[\small Masking, nonmasking stabilization]{%
            \label{fig:mask}
            \includegraphics[height=0.24\textwidth , width=0.29\textwidth]{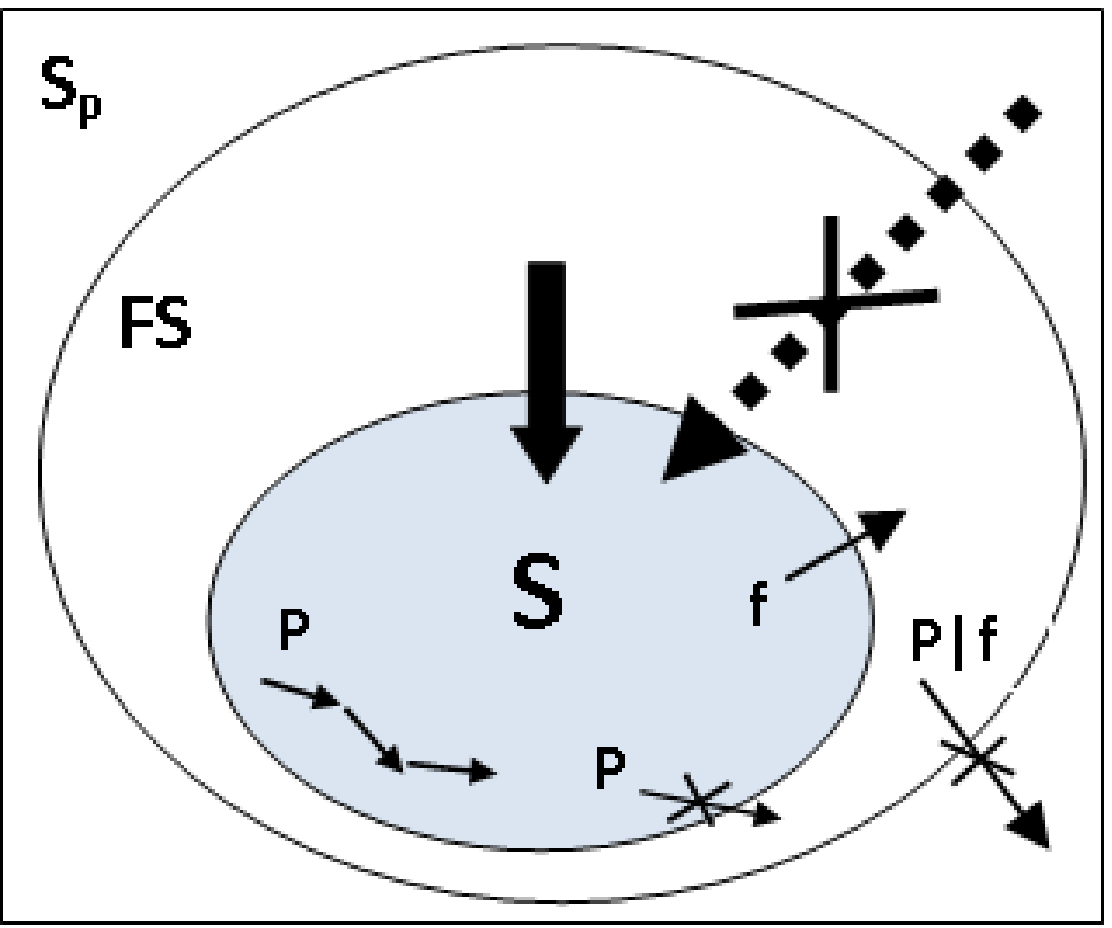}
        }%
        \subfigure[Graceful degradation]{%
           \label{fig:deg}
           \includegraphics[height=0.24\textwidth , width=0.29\textwidth]{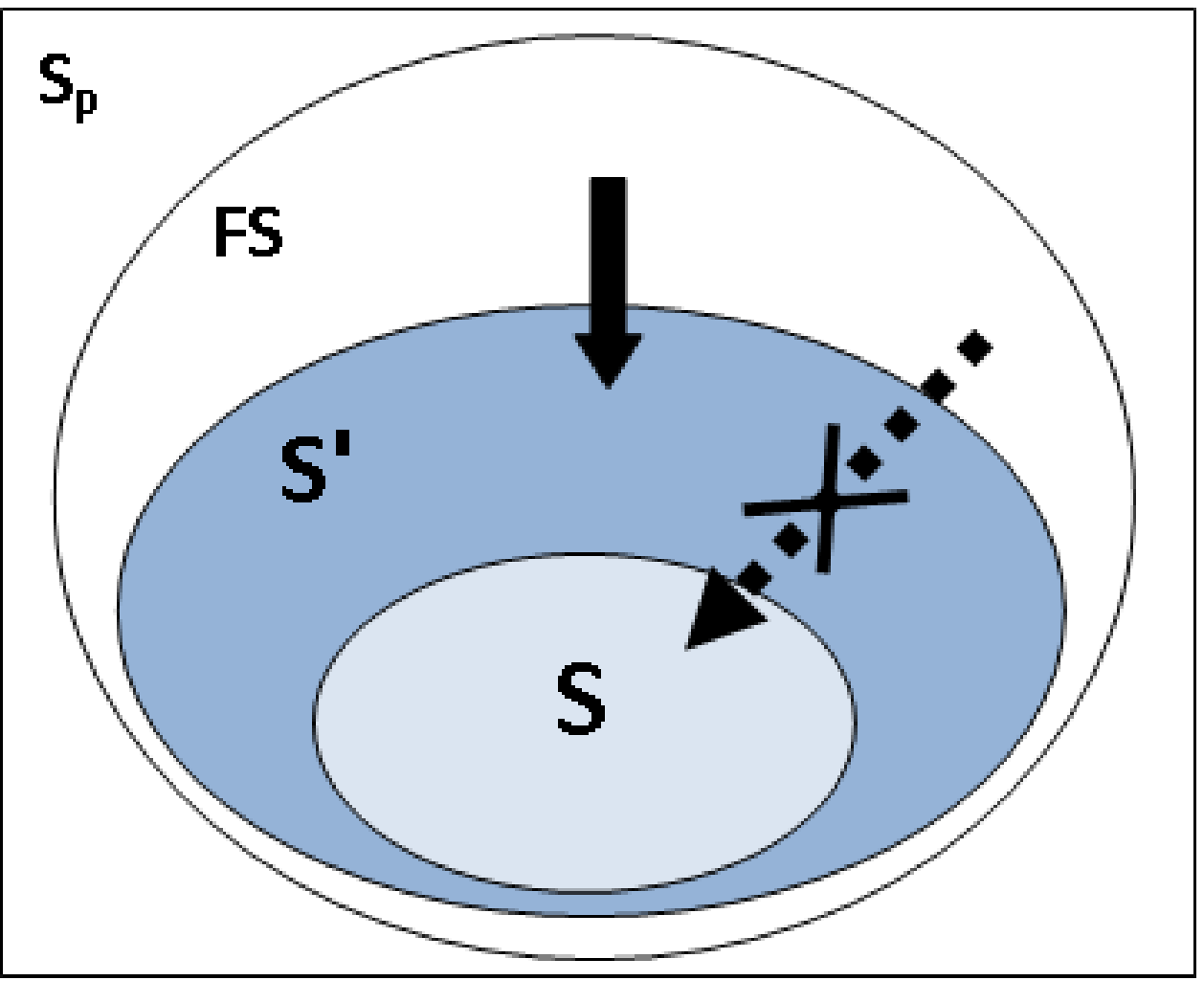}
        } \\ 
        \subfigure[Strict 2-phase recovery]{%
            \label{fig:2ph}
            \includegraphics[height=0.24\textwidth , width=0.29\textwidth]{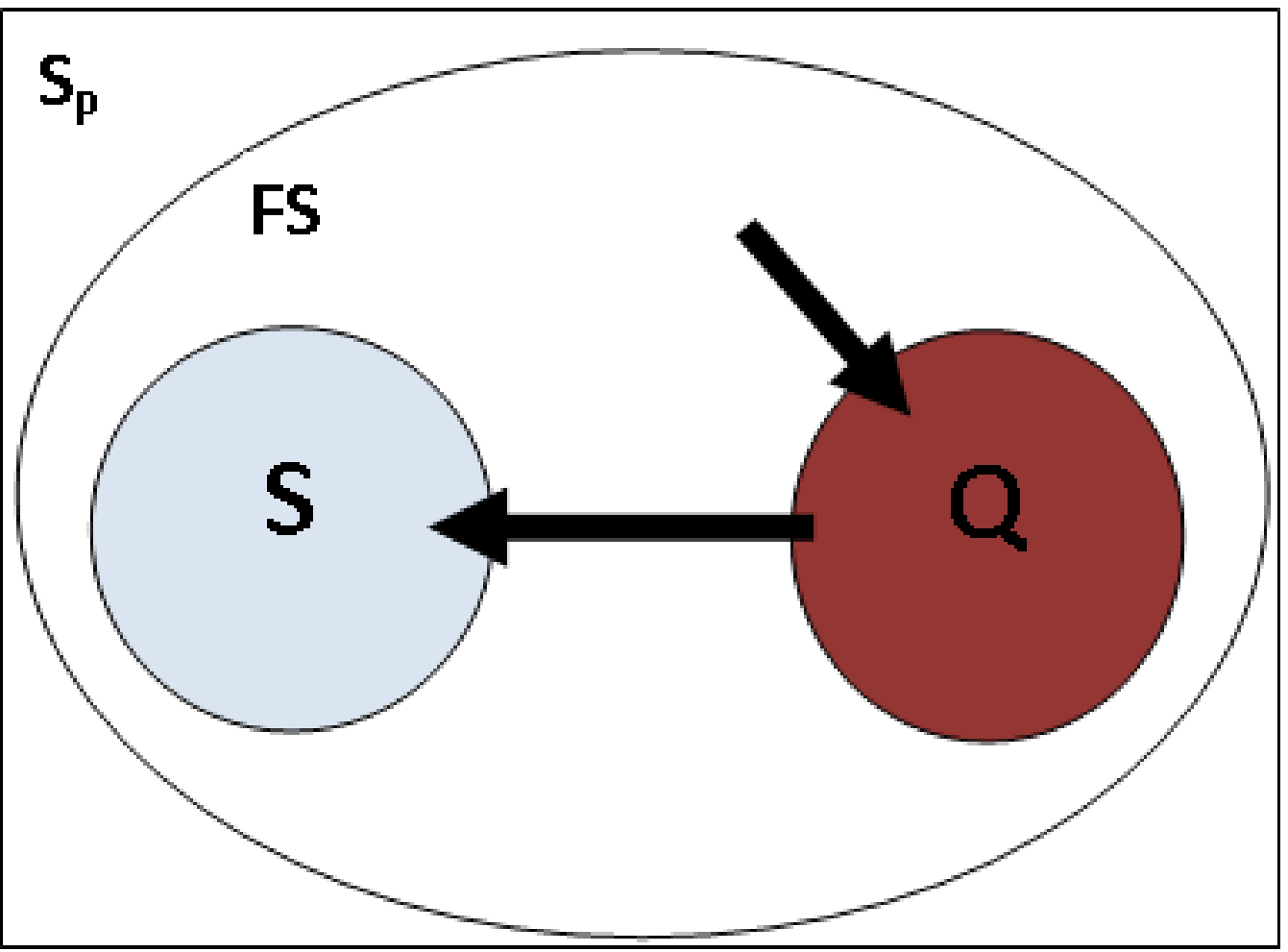}
        } %
        \subfigure[\Ar]{%
           \label{fig:reset}
           \includegraphics[height=0.24\textwidth , width=0.29\textwidth]{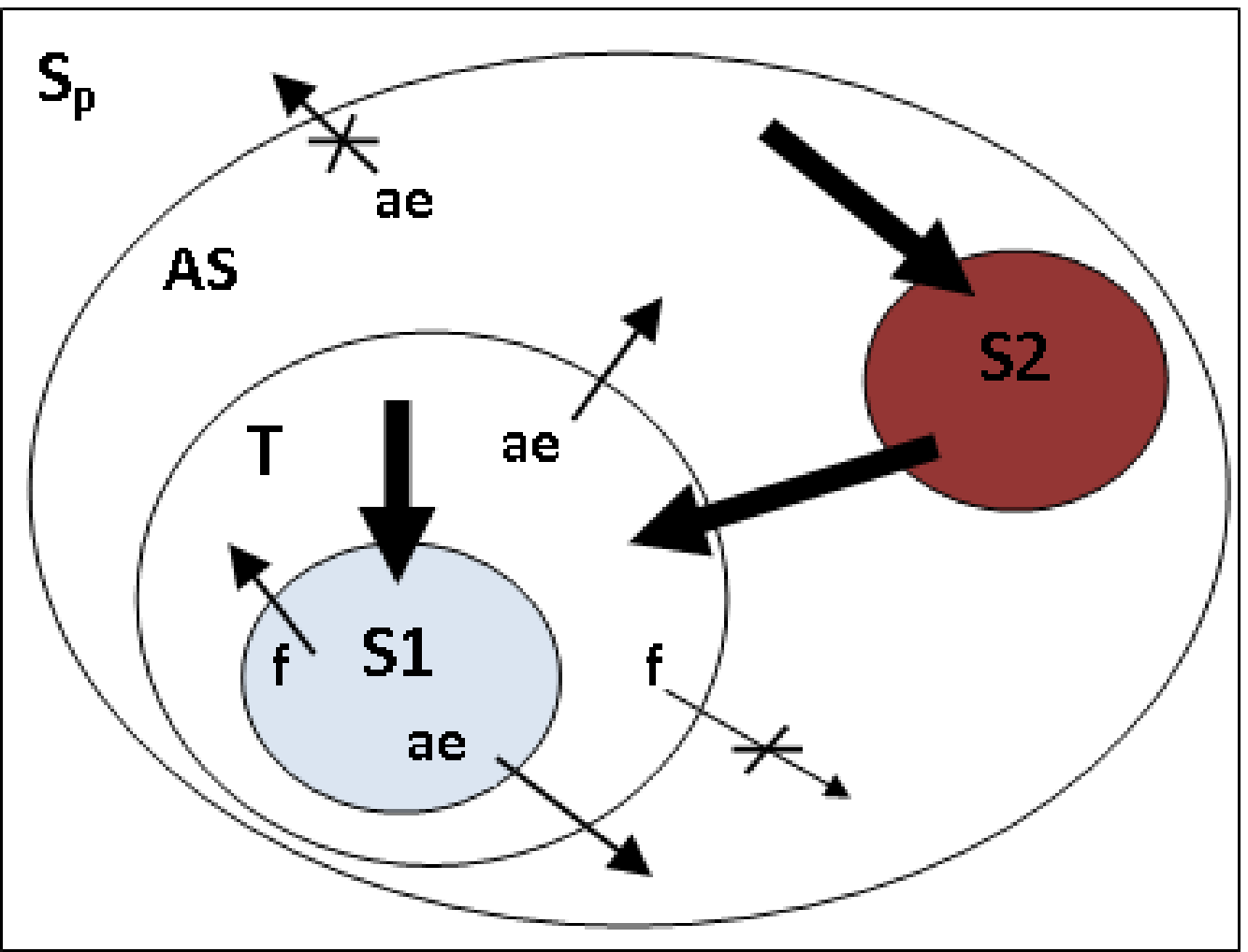}
        } %
        \subfigure{%
           \label{fig:desc}
           \includegraphics[height=0.20\textwidth , width=0.30\textwidth]{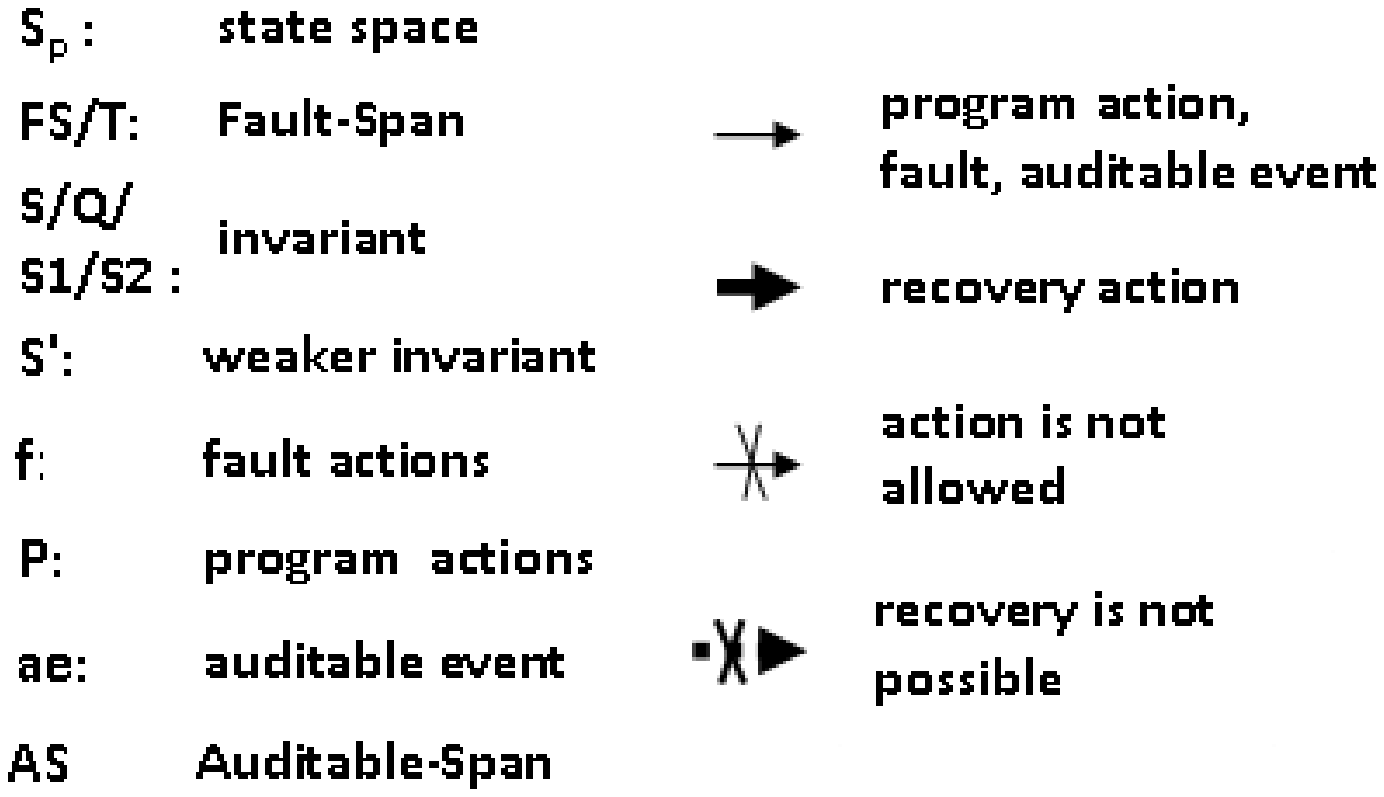}
        } \\ 

    \end{center}
    \caption{%
        Different types of stabilizations.
     }%
    \label{fig:stabilization}

\end{figure*}

{\bf Organization of the paper. } \
The rest of the paper is organized as follows:
In Section~\ref{sec:defs}, we present the preliminary concepts of stabilization and fault-tolerance and introduce the notion of \ar in Section~\ref{sec:ardef}. Section~\ref{sec:ar} explains \ar protocol while the variables are unbounded. We bound these variables in Section \ref{sec:bound} and discuss the necessary changes in our protocol. Section~\ref{sec:related} provides related work, and, finally, we conclude in Section \ref{sec:concl}.

\section{Preliminaries}
\label{sec:defs}

\newcommand{\adv}{{adv}\xspace}
\newcommand{\pure}{pure\xspace}

In this section, we first recall the formal definitions of programs, faults, \aus, and self stabilization adapted from \cite{dij,ag93,wssGouda01}. We then formally define our proposed \ar mechanism.

\begin{definition}[Program]
\label{def:program}

A program $p$ is specified in terms of a set of variables $V$, each of which is associated with a domain of possible values, and a finite set of actions of the form

\begin{tabbing}
\hspace*{3mm} $\langle name \rangle$ \  :: \ $\langle guard \rangle$ \hspace*{2mm} $\longrightarrow$ \hspace*{2mm} $\langle statement \rangle$
\end{tabbing}

where $guard$ is a Boolean expression over program variables, $statement$ updates the program variables.

\end{definition}

For such a program, say $p$, we define the notion of state, state space and transitions next.

\begin{definition}[(Program) State and State Space]
A state of the program is obtained by assigning each program variable a value from its domain. The state space (denoted by $S_p$) of such program is the set of all possible states.
\end{definition}

\begin{definition}[Enabled]
We say that an action {\em g} $\longrightarrow$ $st$ is enabled in state $s$ iff $g$ evaluates to true in $s$.
\end{definition}

\begin{definition}[Transitions corresponding to an action]
The transitions corresponding to an action $ac$ of the form {\em g} $\longrightarrow$ $st$ (denoted by $\delta_{ac}$) is a subset of $S_p \times S_p$ and is obtained as
$\{ (s_0, s_1) | g$ is true in $s_0$ and $s_1$ is obtained executing $st$ in state $s_0$\}.
\end{definition}

\begin{definition}[Transitions corresponding to a program]
The transitions of program $p$  (denoted by $\delta_p$) consisting of actions $ac_1, ac_2, \cdots ac_m$ is

$\delta_p = \bigcup_{i=1}^{m} ac_i \cup$
$\{(s_0, s_0)| ac_1, ac_2 \cdots  ac_m$ are not enabled in $s_0$\}.
\end{definition}

\begin{remark}
Observe that based on the above definition, for any state $s$, there exists at least one transition in $\delta_p$ that originates from $s$. This transition may be of the form $(s, s)$.
\end{remark}

For subsequent discussion, let the state space of program $p$ be denoted by $S_p$ and let its transitions be denoted by $\delta_p$.

\begin{definition}[Computation]
\label{def:comp}

We say that a sequence $\br{s_0, s_1, s_2, ...}$ is a \emph{computation} iff

\begin{itemize}
\item $\forall j \geq 0 :: (s_j, s_{j+1}) \in \delta_p$
\end{itemize}
\end{definition}

\begin{definition}[Closure]
A state predicate $S$ is \emph{closed} in $p$ iff $\forall s_0, s_1 \in S_p :: (s_0 \in S \wedge (s_0, s_1) \in \delta_p) \; \Rightarrow \; (s_1 \in S)$.
\end{definition}

\begin{definition}[Invariant]
A state predicate $S$ is an \emph{invariant} of $p$ iff $S$ is closed in $p$.
\end{definition}

\begin{remark}
Normally, the definition of invariant (legitimate states) also includes a requirement that computations of $p$ that start from an invariant state are correct with respect to its specification. The theory of \ar is independent of the behaviors of the program inside legitimate states. Instead, it only focuses on the behavior of $p$ outside its legitimate states. We have defined the invariant in terms of the closure property alone since it is the only relevant property in the definitions/theorems/examples in this paper.
\end{remark}

\begin{definition}[Convergence]
Let $S$ and $T$ be state predicates of $p$. We say that $T$ \emph{converges} to $S$ in $p$
iff
\begin{itemize}
\item  $S \subseteq T$,
\item $S$ is closed in $p$,
\item $T$ is closed in $p$, and
\item For any computation $\sigma$ (=$\br{s_0, s_1, s_2, ...}$ ) of $p$ if $s_0\in T$ then there exists $l$ such that $s_l \in S$.
\end{itemize}
\end{definition}

\begin{definition}[Stabilization]
We say that program $p$ is \emph{stabilizing} for invariant $S$ iff $S_p$ converges to $S$ in $p$.
\end{definition}

\begin{definition}[Faults]
\emph{Faults} for program $p = \br{S_p, \delta_p}$ is a subset of $S_p \times S_p$; i.e., the faults can perturb the program to any arbitrary state.
\end{definition}

\begin{definition}[Auditable Events]
\emph{Auditable events} for program $p = \br{S_p, \delta_p}$ is a subset of $S_p \times S_p$.
\end{definition}

\begin{remark}
Both faults and \aus are a subset of $S_p \times S_p$. i.e., they both are a set of transitions. However, as discussed earlier, the goal of faults is to model events that are random in nature for which recovery to legitimate states is desired. By contrast, auditable events are deliberate. Additionally, there is a mechanism to detect auditable events and, consequently, we want an auditable restoration technique when \aus occur.
\end{remark}

\begin{definition}[F-Span]
Let $S$ be the invariant of program $p$. We say that a state predicate $T$ is a \emph{f-span} of $p$ from $S$ iff the following conditions are satisfied: (1) $S \subseteq T$, and (2) $T$ is closed in $f \cup \delta_p$.
\end{definition}

\section{Defining Auditable Restoration}
\label{sec:ardef}

In this section, we formally define the notion of auditable restoration. The intuition behind this is as follows: Let $S1$ denote the legitimate states of the program. Let $T$ be a fault-span corresponding to the set of faults $f$. Auditable events perturb the program outside $T$. If this happens, we want to ensure that the system reaches a state in $S2$. Subsequently, we want to restore the system to a state in $S1$.

\begin{definition}[\AR] \label{def:ar}
Let $ae$ and $f$ be \aus and set of faults, respectively, for program $p$.
We say that program $p$ is an \emph{\ar} program with auditable events $ae$ and faults $f$ for \emph{invariant} $S1$ and auditable state predicate $S2$ iff there exists $T$
\begin{itemize}
\item $T$ converges to $S1$ in $p$,
\item $T$ is closed in $\delta_p \cup f$,
\item For any sequence $\sigma$ (=$\br{s_0, s_1, s_2, ...}),w$ s.t.
\end{itemize}
\vspace*{-2mm}
\hspace*{4cm} $s_0 \in T \ \wedge$\\
\hspace*{4cm}$(s_0,s_1) \in ae \ \wedge$\\
\hspace*{4cm}$s_1 \not \in T \wedge$\\
\hspace*{4cm}$(s_m,s_{m+1}) \in \delta_p \cup \delta_f \cup ae \ \wedge$\\
\hspace*{4cm}$m \geq w \Rightarrow (s_m,s_{m+1}) \in \delta_p$

\noindent \hspace*{1cm}$\Rightarrow \exists m,n: ((n > m \geq w) \wedge (s_m \in S2) \wedge (s_n \in S1))$
\end{definition}

\section{\AR for Distributed Programs}
\label{sec:ar}

In this section we explain our proposed \ar protocol for distributed programs.

As mentioned earlier, the \ar protocol consists of several processes. Each process is potentially capable of detecting an auditable event. However, only a subset of processes is capable of initiating the \restoration operation. This captures the intuition that {\em clearing of the auditable event} is restricted to {\em authorized} processes only and can possibly involve human-in-the-loop. For simplicity, we assume that there is a unique process assigned this responsibility and the failure of this process is handled with approaches such as leader election \cite{toplasHuang93}.

The \ar protocol is required to provide the following functionalities: (1) in any arbitrary state where no \au exists, the system reaches a state from where subsequent restoration operation completes correctly, (2) in any arbitrary state where some process detects the \au, eventually all processes are aware of this event, (3) some process in the system can detect that all processes have been aware of the \au, and (4) the process that knows all processes are aware of the \au starts a new restoration operation.

The \ar protocol works as follows: It utilizes a stabilizing silent tree rooted at the \leader process, i.e., it reaches a fixpoint state after building the tree even though individual processes are not aware of reaching the fixpoint. Several tree construction algorithms (e.g., \cite{a96}) can be utilized for this approach. The \ar protocol is superimposed on top of a protocol for tree reconstruction, i.e., it only reads the variables of the tree protocol but does not modify them. The only variables of interest from the tree protocol are $P.j$ (denoting the parent of $j$ in the tree) and $l.j$ (identifying the $id$ of the process that $j$ believes to be the \leader in the tree). Additionally, we assume that whenever the tree action is executed, it notifies the \ar protocol so it can take the corresponding action. Since the tree protocol is silent, this indicates that some tree reconstruction is being done due to faults.

In addition, the program maintains the variable $otsn.j$ and $ctsn.j$. Intuitively, they keep track of the number of auditable events that $j$ has been aware of and the number of auditable events after which the system has been restored by the \leader process. Additionally, each process maintains $st.j$ that identifies its state, $sn.j$ that is a sequence number and $res.j$ that is \{0..1\}. In summary, each process $j$ maintains the following variables:

\begin{itemize}
\item $P.j$, which identifies the parent of process $j$;
\item $l.j$, which denotes the $id$ of the process that $j$ believes to be the \leader;
\item $st.j$, which indicates the state of process $j$. This variable can have four different values: {\em \restore}, {\em \stable}, $\bot$ (read {\em bottom}), or $\top$ (read {\em top}). These values indicate that $j$ is in the middle of restoration after an auditable event has occurred, $j$ has completed its task associated with restoration, $j$ is in the middle of reaching to the auditable state, or $j$ has completed its task associated with reaching to the auditable state, respectively.
\item $sn.j$, which denotes a sequence number;
\item $otsn.j$, as describe above;
\item $ctsn.j$, as described above, and
\item $res.j$, which has the domain \{0..1\}.
\end{itemize}

Observe that in the above program, the domain of $sn.j$, $otsn.j$, and $ctsn.j$ is currently unbounded. This is done for simplicity of the presentation. We can bound the domain of these variables without affecting the correctness of the program. This issue is discussed in Section \ref{sec:bound}.

The program consists of 11 actions. The first action, $\name1$, is responsible for detecting the auditable event. As mentioned above, each process has some mechanism that detects the \aus. If an \au is detected, the process $j$ increments $otsn.j$ and propagates it through the system.

\begin{figure}[t]
\begin{tabbing}
${\bf \name1}::$ \{{\em $j$ detects an \au\}}\\

\=\= \hspace*{0.75cm} \= $\longrightarrow otsn.j:=otsn.j+1 $\\\\


${\bf \name2}::otsn.j < otsn.k$\\

\>\>\> $\longrightarrow otsn.j := otsn.k, \ $ {\bf if} $P.j=j$ {\bf then} $res.j:=0$\\\\


${\bf \name3}:: P.j = j \wedge st.j \neq \bot \wedge \ otsn.j > ctsn.j$\\

\>\>\> $\longrightarrow st.j, sn.j, res.j := \bot, sn.j+1, min(res.j+1,1)$\\\\


${\bf \name4}:: st.(P.j) = \bot \wedge sn.j \neq sn.(P.j) \wedge l.j=l.(P.j)$\\

\>\>\> $\longrightarrow st.j, sn.j, res.j :=$\\

\>\>\> \hspace*{1.5cm} $\bot, sn.(P.j),min(res.j+1,1) $\\\\


${\bf \name5}:: (\forall k: P.k = j: otsn.k=otsn.j  \wedge st.k = \top)$\\

\>\>\>  $\wedge (\forall k: k \in Nbr.j: sn.j=sn.k \wedge l.j=l.k) \wedge st.j = \bot $\\

\>\>\> $\longrightarrow st.j := \top,$\\

\>\>\> \hspace*{6mm} $res.j := min(res.k) \ where \ k \in Nbr.j \cup \{j\}$\\

\>\>\> \hspace*{6mm} {\bf if} $(P.j = j \wedge res.j \neq 1)$ {\bf then}\\
\>\>\> \hspace*{11mm} $st.j, sn.j, res.j := $\\
\>\>\> \hspace*{34mm} $\bot, sn.j+1, min(res.j+1,1)$\\
\>\>\> \hspace*{6mm} {\bf else if} $(P.j=j \wedge res.j=1)$ {\bf then}\\
\>\>\> \hspace*{11mm} $ctsn.j:=otsn.j$\\\\

${\bf \name6}::ctsn.j < ctsn.k \longrightarrow ctsn.j = ctsn.k$\\\\


${\bf \name7}:: P.j = j \wedge \ st.j = \top \wedge ctsn.j=otsn.j \ \wedge$\\
\>\>\> \hspace*{45mm} \{{\em authorized to restore}\}\\
\>\>\> $\longrightarrow st.j, sn.j := \restore, sn.j+1$ \\\\


${\bf \name8}:: st.(P.j) = \restore \wedge sn.j \neq sn.(P.j) \ \wedge$ \\
\>\>\> \hspace*{39mm} $l.j=l.k \wedge otsn.j=ctsn.j$\\
\>\>\> $\longrightarrow st.j, sn.j, res.j :=$\\
\>\>\> \hspace*{15mm} $\restore, sn.(P.j), min(res.j+1,1)$\\\\


${\bf \name9}:: (\forall k: P.k = j: sn.j = sn.k \ \wedge st.k = \stable) \ \; \ \wedge$\\

\>\>\> \hspace*{12mm} $ (\forall k: k \in Nbr.j: sn.j=sn.k \wedge l.j=l.k) \ \wedge $\\

\>\>\> \hspace*{56mm} $st.j = \restore$\\

\>\>\>  $\longrightarrow st.j := \stable,$\\

\>\>\> \hspace*{6mm} $res.j := min(res.k): k \in Nbr.j \cup \{j\}$\\

\>\>\> \hspace*{6mm} {\bf if} $(P.j = j \wedge res.j \neq 1)$ {\bf then}\\
\>\>\> \hspace*{11mm} $st.j, sn.j, res.j := $\\
\>\>\> \hspace*{26mm} $\restore, sn.j+1, min(res.j+1,1)$\\
\>\>\> \hspace*{6mm} {\bf else if} $(P.j=j \wedge res.j=1)$ {\bf then}\\
\>\>\> \hspace*{11mm}  \{{\em \restore is complete}\}\\\\


${\bf \name10}:: \neg \Gd.j \longrightarrow st.j,sn.j :=  st.(P.j),sn.(P.j)$\\\\

${\bf \name11}:: \langle$ {\em any tree correction action that affects process $j$} $\rangle$\\
\>\>\> $ \longrightarrow res.j:=0$

\end{tabbing}
\caption{The \ar protocol.}
\label{fig:ar}
\end{figure}

Actions $\name2$--$\name5$ are for notifying all processes in the system about the \aus detected. Specifically, action $\name2$ propagates the changes of $otsn$ value. We assume that this action is a high priority action and executes concurrently with every other action. Hence, for simplicity, we do not show its addition to the rest of the actions.
When the \leader process is notified of the \au, in action $\name3$, it changes its state to $\bot$ and propagates $\bot$ towards its children. Also, in action $\name3$, the \leader sets its $res$ variable to 1 by using $min(res.j+1,1)$ function. (Instead of setting $res$ value to 1 directly, we utilize this function because it would simplify the bounding of $otsn$, $ctsn$ and $sn$ variables in the next section.)
Action $\name4$ propagates $\bot$ towards the leaves. In action $\name5$, when a leaf receives $\bot$, it changes its state to $\top$ and propagates it towards the \leader.
Consider that, action $\name5$ detects whether all processes have participated in the current $\bot$ wave. This detection is made possible by letting each process maintain the variable $res$ that is true only if its neighbors have propagated that wave. In particular, if process $j$ has completed a $\bot$ wave with $res$ false, then the parent of $j$ completes that wave with the $res$ false. It follows that when the \leader completes the wave with the $res$ true, all processes have participated in that wave.
This action also increments the $ctsn$ value.
However, if the \leader fails when its state is $\bot$ and some process with state $\top$ becomes the new \leader, this assumption will be violated. Action $\name3$ guarantees that, even if the \leader fails, the state of all processes will eventually become $\top$ and the \leader process will be aware of that.

Action $\name6$ broadcasts the changes of the $ctsn$ value so that the other processes can also \restore to their legitimate states. We assume that, similar to action $\name2$, action $\name6$ is also a high priority action and executes concurrently with the other actions.

When the \leader ensures that all processes are aware of the \aus, a human input can ask the \leader to initiate a \restoration wave to recover the system to its legitimate state.
Therefore, in action $\name7$, the \leader initiates a distributed restoration wave, marks its state as {\em \restore}, and propagates the \restoration wave to its children.

When a process $j$ receives a restoration wave from its parent, in action $\name8$, $j$ marks its state as {\em \restore} and propagates the wave to its children. When a leaf process $j$ receives a \restoration wave, $j$ restores its state, marks its state as {\em \stable}, and responds to its parent. In action $\name9$, when the \leader receives the response from its children, the restoration wave is complete.
Action $\name9$, like action $\name5$, detects if all processes participated in the current \restoration wave by utilizing the variable $res$.
To represent action $\name10$, first, we define $\Gd.j$ as follows:


\noindent ${\bf \Gd.j} =$

\noindent \hfill $((st.(P.j) = \restore \wedge st.j = \restore) \Rightarrow sn.j := sn.(P.j) \wedge$\\
\hspace*{0.5cm} $st.(P.j) = \stable  \Rightarrow (st.j=\stable \wedge sn.j=sn.(P.j)) \ \wedge$\\
\hspace*{1.7cm} $(st.(P.j) = \bot  \wedge st.j = \bot) \Rightarrow  sn.j=sn.(P.j) \ \wedge$\\
\hspace*{1.8cm} $st.(P.j) = \top  \Rightarrow (st.j=\top \wedge sn.j=sn.(P.j)))$


Action $\name10$ guarantees the self-stabilization of the protocol by ensuring that no matter what the initials state is, the program can recover to legitimate states from where future \restoration operations work correctly. Finally, if any tree construction algorithm is called to reconfigure the tree and affects process $j$, action $\name11$ resets the $res$ variable of process $j$.


\subsection{Fault Types}
\label{sec:faultar}

We assume that the processes are in the presence of (a) fail-stop faults and (b) transient faults. If a process fail-stops, it cannot communicate with the other processes and a tree correction algorithm needs to reconfigure the tree. Transient faults can perturb all variables (e.g., $sn$, $st$, etc.) except $otsn$ and $ctsn$ values. The corruption of the $otsn$ and $ctsn$ values is tolerated in Section \ref{sec:bound}, where we bound these values. Moreover, we assume that all the faults stop occurring after some time. We use $f$ to denote these faults in the rest of this paper.



\subsection{Proof of Correctness for \AR Protocol}
\label{sec:proofar}

To show the correctness of our protocol, we define the predicates $T$ and $AS$, restoration state predicate $S2$, and invariant $S1$ in the following and use them for subsequent discussions.


\noindent $S1: \forall j,k: ((st.j=\stable \vee st.j=\restore) \wedge \Gd.j \ \wedge$\\
\hspace*{2.5cm} $(P.j$ {\em forms the tree}$) \wedge \ (otsn.j=ctsn.k) \ \wedge$\\
\hspace*{4cm} $ l.j =$ \{{\em \leader of the parent tree}\}$)$

\noindent $T: \forall j,k: otsn.j=ctsn.k \wedge (st.j=\restore \vee st.j=\stable)$

\noindent $AS: max(otsn.j) \geq max(ctsn.j)$

\noindent $S2: \forall j,k: ((st.j = \top \vee st.j = \bot) \ \wedge$\\
\hspace*{2.5cm} $(P.j=j \wedge st.j=\top \Rightarrow otsn.k \geq otsn.j))$


In $S1$ the state of all processes is either $\stable$ or $\restore$ and all $otsn$ and $ctsn$ values are equal. If some faults occur but no \aus, the system goes to $T$ where all $otsn$ and $ctsn$ values are still equal. In this case, the state of the processes cannot be perturbed to $\bot$ or $\top$. If \aus occur, the system goes to $AS$ where the $otsn$ and $ctsn$ values are changed and the $max(otsn)$ is always greater than and equal to $max(ctsn)$. When all the states are changed to $\bot$ or $\top$, all processes are aware of the \aus and the system is in $S2$. Also, the constraint $otsn.k \geq otsn.j$ in the definition of $S2$ means that when the state of the \leader is $\top$, the rest of the processes are aware of the \aus that has caused the \leader to change its $otsn$.

\begin{theorem} \label{thm:t-s1}
Upon starting at an arbitrary state in $T$, in the absence of faults and \aus, the system is guaranteed to converge to a state in $S1$.
\end{theorem}

\begin{proof}
Since there is no \au in the system, action $\name1$ cannot execute and the $otsn$ and $ctsn$ values do not change. As a result, the system remains in $T$ and executing actions $\name8$, $\name9$, and $\name10$ converges the system to $S1$.
\end{proof}


Let $f$ be the faults identified in Section \ref{sec:faultar}. Then, we have:
\begin{theorem} \label{thm:closure-t}
$T$ is closed in actions $(\name2$--$\name11)$ $\cup$ $f$.
\end{theorem}

\begin{proof}
We assume that faults cannot perturb the $otsn$ and $ctsn$ values. Additionally, action $\name1$ does not execute to change the $otsn$ value. This guarantees that actions $\name2$--$\name6$ cannot execute to change the $otsn$ or $ctsn$ values. Hence, when a system is in $T$, in the absence of \aus, it remains in $T$. Moreover, faults cannot perturb the system to a state outside of $T$, thereby closure of $T$.
\end{proof}


\begin{corollary}
Upon starting at an arbitrary state in $T$, in the presence of faults but in the absence of \aus, the system is guaranteed to converge to a state in $S1$.
\end{corollary}


\begin{lemma} \label{thm:notify}
Starting from a state in $T$ where the $otsn$ value of all processes equal $x$, if at least one \au occurs, the system reaches a state where $\forall j: otsn.j \geq x+1$.
\end{lemma}

\begin{proof}
This lemma implies that, if there is an \au in the system and at least one process detects it and increments its $otsn$ value, eventually all processes will be aware of that \au.

When the system is in $T$, all the $otsn$ values are equal to $x$. After detecting an \au, a process increments its $otsn$ value by executing action $\name1$ to $x+1$. Consequently, using action $\name2$, every process gets notified of the \au and increments its $otsn$ value. If some other process detects more \aus in the system, it increments its $otsn$ value and notifies the other processes of those \aus. Hence, the $otsn$ value of all processes is at least $x+1$.
\end{proof}


\begin{theorem} \label{thm:max}
Starting from any state in $AS-T$ where $max(otsn.j) > max(ctsn.j)$ and the \aus stop occurring, the system is guaranteed to reach a state in $S2$.
\end{theorem}

\begin{proof}
When there is a process whose $otsn$ value is greater than all the $ctsn$ values in the system, the state of the system is one of the following:

\begin{itemize}
\item there is at least one process that is not aware of all the \aus occurred in the system, and
\item all processes are aware of all \aus occurred but their states are different.
\end{itemize}

The first case illustrates that all $otsn$ values are not equal. Consequently, action $\name2$ executes and makes all $otsn$ values equal. When the \leader gets notified of the \aus, it initializes a $\bot$ wave by executing action $\name3$. This wave propagates towards the leaves by executing action $\name4$. When a leaf receives the $\bot$ wave, it change its state to $\top$ by executing action $\name5$ and propagates the $\top$ wave towards the \leader. Note that, action $\name2$ executes concurrently with the other actions. Thus, all $otsn$ values will eventually become equal.
The second case shows that all $otsn$ values are equal but the state of the \leader is not changed to $\top$ yet. Therefore, when all $otsn$ values are equal and the \aus stop occurring, the \leader executes action $\name5$ and changes its state to $\top$, thereby reaching $S2$.
Consider that, if a process fails and causes some changes in the configuration of the tree, the $res$ variable of its neighbors would be reset to false and the \leader will eventually get notified of this failure by executing action $\name5$. Thus, the \leader initializes a new $\bot$ by executing action $\name5$.
Moreover, if the \leader fails when its state is $\bot$ and another process, say $j$, whose state is $\top$ becomes the new \leader, the guard of action $\name2$ becomes true since $otsn.j>ctsn.j$. In this case, the new \leader initializes a new $\bot$ wave to ensures that when the state of the \leader is $\top$, the state of all the other processes in the system is also $\top$.
\end{proof}


\begin{theorem} \label{thm:t-s2}
Starting from a state in $T$ and the occurrence of at least one \au, the system is guaranteed to reach $S2$ even if \aus do not stop occurring.
\end{theorem}

\begin{proof}
occurring, at least, one \au, some process, say $j$, detects it and increments its $otsn$ by executing action $\name1$. Following Lemma \ref{thm:notify}, all processes will eventually get notified of the \au. Hence, all $otsn$ values will be equal and, following Theorem \ref{thm:max}, the system will reach a state where the state of all processes is $\top$ or $\bot$. Also if the state of the \leader is $\top$, we can ensure that all processes have been notified of the \au, thereby reaching a state in $S2$. Consider that, even if the \aus continue occurring, the state of processes does not change and the system remains in $S2$.
\end{proof}


\begin{theorem} \label{thm:s1-s1}
Starting from a state in $T$ and in the presence of faults and \aus, the system is guaranteed to converge to $S1$ provided that faults and \aus stop occurring.
\end{theorem}

\begin{proof}
Following Theorem \ref{thm:t-s2}, if faults and \aus occur, the system reaches a state in $S2$ where all $otsn$ values are equal and the \leader process is aware that all processes have been notified of the \aus. In this situation, the $ctsn$ value of the \leader is equal to its $otsn$ value (by executing the statement $ctsn.j:=otsn.j$ in action $\name5$). Consequently, the \leader can start a \restoration wave by executing action$\name7$. Moreover, the other processes can concurrently execute action $\name6$, increase their $ctsn$ value, and propagate the \restoration wave by executing action $\name8$. When all $ctsn$ values are equal, the system is in $T$ and following Theorems \ref{thm:t-s1} and \ref{thm:closure-t}, the system converges to $S1$.
\end{proof}


\begin{observation}
The system is guaranteed to recover to $S2$ in the presence of faults and \aus. As long as \aus continue occurring, the system remains in $S2$ showing that all processes are aware of the \aus. When the \aus stop occurring, the system converges to $S1$, where the system continues its normal execution.
\end{observation}

\section{Bounding \AR Variables}
\label{sec:bound}

In this section, we show how the $otsn$, $ctsn$, and $sn$ values can be bounded while preserving stabilization property.


\subsection{Bounding $otsn$}
In the protocol in Figure \ref{fig:ar}, $otsn$ values continue to increase in an unbounded fashion as the number of auditable events increase. In that protocol, if $otsn$ values are bounded and eventually they are restored to $0$, this may cause the system to lose some auditable events. Furthermore, if $otsn.j$ is reset to $0$ but it has a neighbor $k$ where $otsn.k$ is non-zero, it would cause $otsn.j$ to increase again.

Before we present our approach, we observe that if the auditable events are too frequent, restoring the system to legitimate states may never occur. This is due to the fact that if the \leader process attempts to restore the system to legitimate states then that action would be {\em canceled} by new auditable events. Hence, restoring the system to legitimate states can occur only after auditable events stop. However, if auditable events occur too frequently for some duration, we want to ensure that the $otsn$ values still remain bounded.

Our approach is as follows: We change the domain of $otsn$ to be $N^2+1$, where $N$ is the number of processes in the system. Furthermore, we change action $\name1$ such that process $j$ ignores the detectable events if its neighbors are not aware of the recent auditable events that it had detected. Essentially, in this case, $j$ is {\em consolidating} the auditable events. Furthermore, we change action $\name2$ by which process $j$ detects that process $k$ has detected a new auditable event. In particular, process $j$ concludes that process $k$ has identified a new auditable event if $otsn.k$ is in the range $[otsn.j\oplus 1 \cdots otsn.j\oplus N]$, where $\oplus$ is modulo $N^2+1$ addition. Moreover, before $j$ acts on this new auditable event, it checks that its other neighbors have caught up with $j$, i.e., their $otsn$ value is in the range $[otsn.j \cdots otsn.j\oplus N]$. Finally, we add another action where $otsn.j$ and $otsn.k$ are far apart, i.e., $otsn.j$ is not in the range $[otsn.k \ominus N \cdots otsn.k \oplus N]$. Thus, the revised and new actions are $\bname1$, $\bname2$, and $\bname12$ in Figure \ref{fig:bounded-ar}.

We have utilized these specific actions so that we can benefit from previous work on asynchronous unison \cite{cfg} to bound the $otsn$ values. Although the protocol in \cite{cfg} is designed for clock synchronization, we can utilize it to bound the $otsn$ values. In particular, the above actions are same (except for the detection of new auditable events) as that of \cite{cfg}. Hence, based on the results from \cite{cfg}, we can observe that if some process continues to detect auditable events forever, eventually, the system would converge to a state where the $otsn$ values of any two neighboring processes differ by at most $1$.

\begin{theorem} \label{thm:async-unison}
Starting from an arbitrary state, even if the auditable events occur at any frequency, the program in Figure \ref{fig:bounded-ar} converges to states where for any two neighbors $j$ and $k$: $otsn.j = otsn.k\ominus 1$, $otsn.j = otsn.k$, or $otsn.j = otsn.k\oplus 1$.
\end{theorem}

\begin{proof}
Proof follows from \cite{cfg}.
\end{proof}

\begin{theorem} \label{thm:otsn-equal}
Starting from an arbitrary state, if the auditable events stop occurring, the program in Figure \ref{fig:bounded-ar} converges to states where for any two neighbors $j$ and $k$: $otsn.j = otsn.k$.
\end{theorem}

\begin{proof}
As we mentioned in Theorem \ref{thm:async-unison}, the system is guaranteed to converge to a state where for any two neighbors $j$ and $k$ either $otsn.j = otsn.k\ominus 1$, $otsn.j = otsn.k$, or $otsn.j = otsn.k\oplus 1$.
Now, $otsn$ values of any two processes (even if they are not neighbors) differ by at most $N-1$.
In other words, there exists $a$ and $b$ such that for some processes $j$ and $k$, $otsn.j = a$ and  $otsn.k = b $, where $b$ is in the range $[a \cdots a \oplus N]$ and the $otsn$ values of remaining processes are in the range $[a \cdots b]$.

Hence, processes are not far apart each other and the action $\bname12$ cannot execute. In addition, action $\bname1$ cannot execute since the auditable events have stopped occurring. Hence, by executing action $\bname2$, each process increases its $otsn$ such that the $otsn$ values will be equal to $b$ for all processes.
\end{proof}

\begin{corollary}
Under the assumption that a process does not detect new auditable event until it is restored to legitimate states, we can guarantee that $otsn$ values will differ by no more than 1.
\end{corollary}

\begin{figure}[t]
\begin{tabbing}
${\bf \bname1}::$ \{{\em $j$ detects an \au\}}\\
\=\= \hspace*{0.75cm} \=  $\forall k: otsn.k \in [otsn.j \cdots otsn.j \oplus N]$\\
\=\= \hspace*{0.75cm} \= $\longrightarrow otsn.j:=otsn.j \oplus 1 $\\\\


${\bf \bname2}:: \forall k: otsn.k \in [otsn.j \cdots otsn.j \oplus N] \ \wedge$\\
\>\>\> \hspace*{1.7cm} $\exists k: otsn.k \in [otsn.j \oplus 1 \cdots otsn.j \oplus N]$\\
\>\>\> $\longrightarrow otsn.j := otsn.j \oplus 1$\\\\


${\bf \bname3}:: P.j = j \wedge st.j \neq \bot \wedge \ otsn.j \neq ctsn.j$\\

\>\>\> $\longrightarrow st.j, sn.j, res.j := \bot, sn.j+1, min(res.j+1,1)$\\\\


${\bf \bname4}:: st.(P.j) = \bot \wedge sn.j \neq sn.(P.j) \wedge l.j=l.(P.j)$\\

\>\>\> $\longrightarrow st.j, sn.j, res.j := \bot, sn.(P.j),res.(P.j) $\\\\


${\bf \bname5}::{\bf \name5}$\\\\


${\bf \bname6}:: ctsn.j \neq ctsn.(P.j) \longrightarrow ctsn.j := ctsn.(P.j)$\\\\


${\bf \bname7}::{\bf \name7}$\\\\


${\bf \bname8}:: st.(P.j) = \restore \wedge sn.j \neq sn.(P.j) \ \wedge$ \\
\>\>\> \hspace*{39mm} $l.j=l.k \wedge otsn.j=ctsn.j$\\
\>\>\> $\longrightarrow st.j, sn.j, res.j := \restore, sn.(P.j), res.(P.j)$\\\\

${\bf \bname9}$--${\bf \bname10}::{\bf \name9}$--${\bf \name10}$\\\\


${\bf \bname11}:: \langle$ {\em any tree correction action that affects process $j$} $\rangle$\\
\>\>\> $ \longrightarrow res.j:=-1$\\\\


${\bf \bname12}:: (otsn.j \not \in [otsn.k \ominus N \cdots otsn.k \oplus N]) \ \wedge$\\
\>\>\> \hspace*{2.6cm} $(otsn.j > otsn.k) \longrightarrow otsn.j := 0$


\end{tabbing}
\caption{Bounded \ar protocol.}
\label{fig:bounded-ar}
\end{figure}



\subsection{Bounding $ctsn$}

The above approach bounds the $otsn$ value. However, the same approach cannot be used to bound $ctsn$ value. This is due to the fact that the value to which $ctsn$ converges may not be related to the value that $otsn$ converges to. This is unacceptable and, hence, we use the following approach to bound $ctsn$.

First, in action $\name6$, the guard $ctsn.j > ctsn.k$ needs to be replaced by $ctsn.j \neq ctsn.(P.j)$. Thus, the new action $\bname6$ is shown in Figure \ref{fig:bounded-ar}.


Second, we require to replace the notion of {\em greater than} by {\em not equal} in all the actions of Figure \ref{fig:ar} since we are bounding the $otsn$ and $ctsn$ values. Hence, we change action $\name3$ to $\bname3$ in the program in Figure \ref{fig:bounded-ar}.

With these changes, starting from an arbitrary state, after the auditable events stop, the system will eventually reach to states where all $otsn$ values are equal (cf. Theorem \ref{thm:otsn-equal}). Subsequently, if $otsn$ and $ctsn$ values of the \leader process are different, it will execute action $\name7$ to restore the system to an auditable state. Then, the \leader process will reset its $ctsn$ value to be equal to $otsn$ value. Finally, these values will be copied by other processes using action $\bname6$. Hence, eventually all $ctsn$ values will be equal.

\begin{theorem}
Starting from an arbitrary state, if the auditable events stop occurring, the program in Figure \ref{fig:bounded-ar} reaches to states where for any two neighbors $j$ and $k$: $otsn.j = otsn.k=ctsn.j=ctsn.k$.
\end{theorem}

\begin{proof}
According to Theorem \ref{thm:otsn-equal}, all $otsn$ values will eventually be equal. Moreover, when the \leader ensures that all processes are aware of the auditable events, it executes action $\name5$ and updates its $ctsn$ value by its $otsn$ value. Consequently, when the rest of processes detect this change, they execute action $\bname7$ and update their $ctsn$ values. Hence, eventually all $ctsn$ values will be equal.
\end{proof}

Finally, we observe that even with these changes if a single auditable event occurs in a legitimate state (where all $otsn$ and $ctsn$ values are equal) then the system would reach a state in the auditable state, i.e., Theorem \ref{thm:t-s2} still holds true with this change.


\subsection{Bounding $sn$}

Our goal in bounding $sn$ is to only maintain $sn \ mod \ 2$ with some additional changes. To identify these changes, first, we make some observations about how $sn$ values might change during the computation in the presence of faults such as process failure but in the absence of transient faults.

Now, consider the case where we begin with a legitimate state of the \ar protocol where all $sn$ values are equal to $x$. At this time if the \leader process executes actions $\name3$ or $\name7$ then $sn$ value of the \leader process will be set to $x+1$. Now, consider the $sn$ values of processes on any path from the \leader process to a leaf process. It is straightforward to observe that some initial processes on this path will have the $sn$ value equal to $x+1$ and the rest of the processes on this path (possibly none) will have the $sn$ value equal to $x$. Even if some processes fail, this property would be preserved in the part of the tree that is still connected to the \leader process. However, if some of the processes in the subtree of the failed process (re)join the tree, this property may be violated.
In the program in Figure \ref{fig:ar}, where $sn$ values are unbounded, these newly (re)joined processes can easily identify such a situation. However, if processes only maintain the least significant bit of $sn$, this may not be possible. Hence, this newly rejoined process should force the \leader process to redo its task for recovering the system to either auditable state (i.e., $S2$) or to legitimate states (i.e., $S1$). As described above there can be at most two possible values of $sn$ in the tree that is connected to the \leader process. Hence, it suffices that the newly rejoined process aborts those two computations. We can achieve this by actions $\bname3$, $\bname4$, $\bname8$, and $\bname11$ in the program in Figure \ref{fig:bounded-ar}.

In this program, consider the case where one auditable event occurs and no other auditable event occurs until the system is restored to the legitimate states. In this case, for any two processes $j$ and $k$, $otsn.k$ will be either $otsn.j$ or $otsn.j \oplus 1$. In other words, we redefine predicates $S2$ and $AS$ in the following. The definitions of $S1$ and $T$ remain unchanged.

\noindent $AS': \forall j,k: otsn.k \in [max(ctsn.j) \cdots max(ctsn.j) \oplus 1]$

\noindent $S2': \forall j,k: ((st.j = \top \vee st.j = \bot) \ \wedge$\\
\hspace*{0.4cm} $(P.j=j \wedge st.j=\top \Rightarrow otsn.k \in [otsn.j \cdots otsn.j \oplus N])$

\begin{theorem} \label{thm:s1-s1-1}
Starting from an arbitrary state in $T$, if exactly one \au occurs, the system is guaranteed to reach a state in $S2$ and then converge to $S1$ provided that faults stop occurring. Moreover, the states reached in such computation are a subset of $AS'$.
\end{theorem}

\begin{proof}
If exactly one \au occurs, the system is guaranteed to reach $S2$ and then it converges to $S1$ following the explanations above.
\end{proof}

We can easily extend the above theorem to allow upto $N$ auditable events before the system is restored to its legitimate state. The choice of $N$ indicates that each process detects the auditable event at most once before the system is restored to its legitimate states. In other words, a process ignores auditable events after it has detected one and the system has not been restored corresponding to that event. In this case, the constraint $AS'$ above needs to be changed to:

\noindent $AS'': \forall j,k: otsn.k \in [max(ctsn.j) \cdots max(ctsn.j) \oplus N]$

\begin{theorem} \label{thm:s1-s1-n}
Starting from an arbitrary state in $T$, if upto $N$ \aus occur, the system is guaranteed to reach a state in $S2$ and then converge to $S1$ provided that faults stop occurring. Moreover, the states reached in such computation are a subset of $AS''$.
\end{theorem}

\begin{proof}
Since each process detects at most one auditable event before the restoration, executing actions $\bname1$ and $\bname2$, all $otsn$ values will become equal and the system reaches $S2$. Consequently, the \leader initialize a restoration wave and the system converges to $S1$.
\end{proof}

Since the domain of $otsn$ is bounded, it is potentially possible that $j$ starts from a state where $otsn.j$ equals $x$ and there are enough auditable events so that the value of $otsn.j$ rolls over back to $x$. In our algorithm, in this case, some auditable events may be lost. We believe that this would be acceptable for many applications since the number of auditable events being so high is highly unlikely. Also, the domain of $otsn$ and $ctsn$ values can be increased to reduce this problem further. This problem can also be resolved by ensuring that the number of events detected by a process within a given time-span is bounded by allowing the process to ignore frequent auditable events.

We note that theoretically the above assumption is not required. The basic idea for dealing with this is as follows: Each process maintains a bit  $changed.j$ that is set to true whenever $otsn$ value changes. Hence, even if $otsn$ value rolls over to the initial value, $changed.j$ would still be true. It would be used to execute action $\bname3$ so that the system would be restored to $S2$ and then to $S1$. In this case, however, another computation would be required to reset $changed.j$ back to false. The details of this protocol are outside the scope of this paper.

Finally, we note that even if the $otsn$ values roll over or they are corrupted to an arbitrary value, the system will still recover to states in $S1$. In particular, starting from an arbitrary state, after faults and auditable events stop, the system will reach a state where $otsn$ values are equal. In this case, depending upon the $ctsn$ values, either the system will be restored to the auditable states (by actions $\bname3$-$\bname6$) or to legitimate states (by action $\bname10$). In the former case, the system will be restored to $S1$ subsequently. Hence, we have the following theorem.

\begin{theorem}
Auditable restoration program is stabilizing, i.e., starting from an arbitrary state in $T$, after \aus and faults stop, the program converges to $S1$.
\end{theorem}

\begin{proof}
Following Theorem \ref{thm:s1-s1-n} and explanations above, after \aus and faults stop, the program converges to $S1$.
\end{proof}

\section{Related Work}
\label{sec:related}

{\bf Stabilizing Systems. } \
There are numerous approaches for recovering a program to its set of legitimate states. Arora and Gouda's distributed reset technique introduced in \cite{tcAroraG94} is directly related to our work and ensures that after completing the reset, every process in the system is in its legitimate states. However, their work does not cover faults (e.g., process failure) during the reset process.
In \cite{reset}, we extended the distributed reset technique so that if the reset process is initialized and some faults occur, the reset process works correctly.
In \cite{dcKatz93}, Katz and Perry showed a method called global checking and correction to periodically do a snapshot of the system and reset the computation if a global inconsistency is detected. This method applies to several asynchronous protocols to convert them into their stabilizing equivalent, but is rather expensive and insufficient both in time and space.

Moreover, in nonmasking fault-tolerance (e.g., \cite{jhsnArora96, a96}), we have the notion of fault-span too (similar to $T$ in Definition \ref{def:ar}) from where recovery to the invariant is provided. Also, in nonmasking fault-tolerance, if the program goes to $AS-T$, it {\em may} recover to $T$. By contrast, in \ar, if the program reaches a state in $AS-T$, it is required that it first restores to $S2$ and then to $S1$. Hence, \ar is stronger than the notion of nonmasking fault-tolerance.


\Ar can be considered as a special case of nonmasking-failsafe multitolerance (e.g., \cite{tosemEbnenasirK11}), where a program that is subject to two types of faults $F_f$ and $F_n$ provides (i) failsafe fault tolerance when $F_f$ occurs, (ii) nonmasking tolerance in the presence of $F_n$, and (iii) no guarantees if both $F_f$  and $F_n$  occur in the same computation.

{\bf Tamper Evident Systems. }
These systems \cite{aegis,ltmlbmtamper00,tamper2} use an architecture to protect the program from external software/hardware attacks. An example of such architecture, AEGIS \cite{aegis}, relies on a single processor chip and can be used to satisfy both integrity and confidentiality properties of an application.
AEGIS is designed to protect a program from external software and physical attacks, but did not provide any protection against side-channel or covert-channel attacks. %
In AEGIS, there is a notion of recovery in the presence of a security intruder where the system recovers to a `less useful' state where it declares that the current operation cannot be completed due to security attacks. However, the notion of fault-tolerance is not considered. In the context of  Figure \ref{fig:ar}, in AEGIS, $AS-T$ equals $\neg S1$, and $S2$ corresponds to the case where tampering has been detected.

{\bf Byzantine-Tolerant Systems. } \
The notion of Byzantine faults \cite{lsp} has been studied in a great deal in the context of fault-tolerant systems. Byzantine faults capture the notion of a malicious user being part of the system. Typically, Byzantine fault is mitigated by having several replicas and assuming that the number of malicious replicas is less than a threshold (typically, less than $\frac{1}{3}$rd of the total replicas). Compared with Figure \ref{fig:ar}, $T$ captures the states reached by Byzantine replicas. However, no guarantees are provided outside $T$.

{\bf Byzantine-Stabilizing Systems. } \
The notion of Byzantine faults and stabilization have been combined in \cite{na02,dw95,m06}. In these systems, as long as the number of Byzantine faults is below a threshold, the system provides the desired functionality. In the event the number of Byzantine processes increases beyond the threshold temporarily, the system eventually recovers to legitimate state. Similar to systems that tolerate Byzantine faults, these systems only tolerate a specific malicious behavior performed by the adversary. It does not address active attacks similar to that permitted by Dolev-Yao attacker \cite{dy83}.

In all the aforementioned methods the goal is to restore the system to the legitimate states. In our work, if some auditable events occur, we do not recover the system to the legitimate states. Instead, we recover the system to restoration state where all processes are aware of the auditable events occurred. Moreover, in \cite{tcAroraG94}, if a process requests a new reset wave while the last reset wave is still in process, the new reset will be ignored. Nevertheless, in our work, we cannot ignore auditable events and all processes will eventually be aware of these events and the system remains in the auditable state as long as the \aus continue occurring. Also the aforementioned techniques cannot be used to bound $otsn$ value in our protocol. 
\vspace*{-2mm}
\section{Conclusion and Future Work}
\label{sec:concl}	

In this paper, we presented an algorithm for {\it \ar} of distributed systems. This problem is motivated in part by the need for dealing with conflicting requirements. Examples of such requirements include cases where access must be restricted but in some access entirely preventing access is less desirable than some unauthorized access. This also allows one to deal with systems where resistance to tampering/unauthorized access is based on user norms or a legal threat as opposed to a technical guarantee that tampering cannot occur. In other words, such systems provide a {\em In-case-of-emergency-break-glass} method for access. By design, these access methods are detectable and called auditable events in our work.

While such an approach often suffices for centralized systems, it is insufficient for distributed systems. In particular, in a distributed system, only one process will be aware of such auditable events. This is unacceptable in a distributed system. Specifically, it is essential that all (relevant) processes detect this auditable event so that they can provide differential service if desired. We denote such states as {\em auditable states}.

Auditable events also differ from the typical {\em In-case-of-emergency-break-glass} events. Specifically, the latter are one-time events and cannot repeat themselves. By contrast, auditable events provide the potential for multiple occurrences. Also, they provide a mechanism for {\em clearing} these events. However, the clearing must be performed after all processes are aware of them and after initiated by an authorized process.

Our program guarantees that after auditable events occur the program is guaranteed to reach an auditable state where all processes are aware of the auditable event and the authorized process is aware of this and can initiate the restoration (a.k.a. clearing) operation. The recovery to auditable state is guaranteed even if it is perturbed by finite number of auditable events or faults.
It also guarantees that no process can begin the task of restoration until recovery to auditable states is complete.
Moreover, after the authorized process begins the restoration operation, it is guaranteed to complete even if it is perturbed by a finite number of faults. However, it will be aborted if it is perturbed by new auditable events.

Our program is stabilizing in that starting from an arbitrary state, the program is guaranteed to reach a state from where future auditable events will be handled correctly. It also utilizes only finite states, i.e., the values of all variables involved in it are bounded.

We are currently investigating the design and analysis of \ar of System-on-Chip (SoC) systems in the context of the IEEE SystemC language. Our objective here is to design systems that facilitate reasoning about what they do and what they do not do in the presence of \aus. Second, we plan to study the application of \ar in game theory (and vice versa).

\vspace*{-2mm}

\bibliographystyle{plain}
\bibliography{bibliography}

\end{document}